\def\be{\begin{equation}}
	\def\ee{\end{equation}}
\def\ba{\begin{array}}
	\def\ea{\end{array}}

\def\mathbi#1{\text{\em #1}}

\documentclass[prl,showpacs,twocolumn,amsmath]{revtex4}
\usepackage{amsmath}
\usepackage{amsfonts}
\usepackage{mathrsfs}
\usepackage{amssymb}
\usepackage{pifont}
\usepackage{epsfig,subfigure,dsfont,amsthm,amsbsy,mathrsfs,amscd}
\usepackage{epstopdf}
\usepackage{bbm}
\usepackage{color}
\def\qed{\leavevmode\unskip\penalty9999 \hbox{}\nobreak\hfill
	\quad\hbox{\leavevmode  \hbox to.77778em{%
			\hfil\vrule   \vbox to.675em%
			{\hrule width.6em\vfil\hrule}\vrule\hfil}}
	\par\vskip3pt}
\usepackage{leftidx}
\newtheorem{theorem}{Theorem}
\newtheorem{corollary}{Corollary}
\newtheorem{lemma}{Lemma}
\newtheorem{observation}{Observation}
\begin{document}
\title{\large\bf Notes on detection and measurement of quantum coherence}

\author{Yiding Wang$^{1}$, Tinggui Zhang$^{1, \dag}$}
\affiliation{ ${1}$ School of Mathematics and Statistics, Hainan Normal University, Haikou, 571158, China \\
	$^{\dag}$ Correspondence to tinggui333@163.com}

\bigskip
\bigskip

\begin{abstract}
Quantum coherence is one of the most basic characteristics of quantum mechanics. Here we give some methods to detect and measure quantum coherence. Firstly, we propose a coherence criterion without full quantum state tomography based on partial transposition. Moreover, we present a coherent nonlinear detection strategy from witnesses, in which we find that for some coherent states, normal witness detection fails but our nonlinear detection succeeds. In addition, we prove that when the nonlinear detection on the two copies of the coherent state fails, the nonlinear detection on the three copies may be successful. Finally, due to the difficulty in calculating robustness of coherence for general states, we introduce a lower bound for coherent robustness based on the witness operator, and after comparing our lower bound with the currently known lower bound, one show that our lower bound is better. Coherence is believed to play a crucial role in quantum information tasks, making the detection and quantization of coherence particularly significant. Therefore, these results help to open up new avenues for advancement in quantum theory.
\end{abstract}

	 \maketitle

\section{I. Introduction}

Quantum coherence, the most fundamental feature signifying quantumness in a single system and underpinning all quantum correlations in composite systems \cite{rhph,mbc,bcps}, plays a significant role in quantum theory. Coherence has many applications in various quantum technologies, such as quantum metrology \cite{vgsl,vgsll}, quantum computing \cite{mh}, nanoscale thermodynamics \cite{ggmp,mldj,vngg,gfjd} and biological systems \cite{sl,sfhm,erra}. Due to the importance of coherence in various quantum information-processing tasks, the detection and quantization of coherence have become particularly crucial and critical.

Consider a $d$-dimensional Hilbert space $\mathcal{H}$ with computational basis $\mathcal{B}:=\{|i\rangle|\,i=1,2,...,d\}$, a quantum state is defined as an incoherent state under the basis $\mathcal{B}$ \cite{bcp} if and only if its density matrix can be expressed as 
\begin{equation}\label{e1}
	\rho=\sum_{i=1}^{d}p_i|i\rangle\langle i|,
\end{equation}
where $0\leq p_i\leq1$ and $\sum_{i=1}^{d}p_i=1$. We denote that IC is the set of all incoherent states. In terms of coherence detection, an experimentally implementable method to detect coherence is by using coherence witness \cite{ctmm,pcbn,wtwy,hral,nzg,ddxy,bhws,zzy,wdm,lxfz}. Coherence witness was first introduced by Napoli et al. \cite{ctmm}, in which they provided the definition of coherence witness and based on this witness, presented a lower bound on the robustness of coherence. In \cite{hral}, the authors define another coherent witness as an observable which mean value vanishes for all incoherent states but nonzero for some coherent states, and its mean value can quantitatively reveal the amount of coherence. The authors in \cite{nzg} illustrate that the normal witness could mistake an incoherent state as a coherent state due to the inaccurate settings of measurement bases and propose a measurement-device-independent coherence witness scheme without any assumptions on the measurement settings. Moreover, Wang et al. show that a coherence witness is optimal if and only if all of its diagonal elements are zero \cite{bhws}. In \cite{zzy}, two coherence witnesses are given based on their witness-constructing manner, which are also used to estimate the robustness of coherence, $l_1$-norm and $l_2$-norm of coherence measure. Recently, Li et al. \cite{lxfz} prove that if the trace of the observable is a known quantity, the range of coherence-detection capabilities can be extended, while a series of coherent criteria have been established based on this prior knowledge. However, the construction of witnesses is not a straightforward task and relies on making use of specific knowledge in the state to be detected.

For quantization of quantum coherence, significant advancements of which have been made based on many innovative approaches \cite{tbmc,asus,xyhz,awdy,ctmm,pcbn,xqtg,kfbu,zxjs,mrtr,ywss,fbhk,dhyl,jxlh,llqw,bkb,rgps,ysnl,llss,dhyc}. In which many coherence measures have been presented, such as robustness of coherence \cite{ctmm,pcbn}, geometric measure of coherence \cite{asus}, and distance-based coherence measure \cite{tbmc,zxjs}. In \cite{xqtg}, the authors propose a valid quantum coherence measure called coherence concurrence and show that any degree of coherence with respect to some reference basis can be converted to entanglement via incoherent operations. Bu et al. \cite{kfbu} introduce a new coherence measure based on maximum relative entropy, and prove that the smooth maximum relative entropy of coherence provides a lower bound of one-shot coherence cost, and the maximum relative entropy of coherence is equivalent to the relative entropy of coherence in the asymptotic limit for a suitable smooth maximum relative entropy. In addition, incoherent states can be seen as generated by a von Neumann measurement, inspired by this, positive-operator-valued measures(POVMs)-based coherence measures and POVM-incoherent operations were established in \cite{fbhk}. Ray et al. in \cite{rgps} provide a lower bound on the quantum coherence for an arbitrary quantum state by using a noncommutativity estimator of an arbitrary observable of subunit norm. The coherence measures are also investigated in terms of Fisher information \cite{dhyc}. More recently, A. Budiyono et al. illustrate that the $l_1$-norm of the Kirkwood-Dirac quasiprobability over an incoherent reference basis and a second basis can be used to quantify quantum coherence \cite{bhkd}. Based on the quantum optimal transport cost, the author in \cite{xshi} establish a coherence measure via convex-roofs. 

In this work, we first introduce an experimentally friend coherence criterion without full quantum state tomography [Section 2]. Secondly, we propose coherence nonlinear detection to more effectively utilize existing coherence witnesses [Section 3]. Moreover, robustness of coherence is a useful measure for coherence, but for general coherent states, coherent robustness is difficult to calculate. So we provide a lower bound for coherence robustness from the witnesses and demonstrate its effectiveness by comparing it with other currently known lower bounds [Section 4]. We summarize and discuss our conclusions in the last section.

\section{II. Coherence criterion without any prior knowledge}

Partial transposition was initially introduced as a tool for detecting entanglement \cite{ap}, in which Peres proposed the positive partial transpose (PPT) criterion. Given a bipartite state $\rho=(\rho)_{ij,kl} \in H_A\otimes H_B$, where $i,k=1,2,\cdots,M$ and $j,l=1,2,\cdots,N$; $M$ and $N$ are the dimensions of subsystem $A$ and subsystem $B$, respectively. If the state is separable, then the matrix $\rho^{\tau_B}$ obtained from partial transpose with respect to subsystem $B$ is still positive semidefinite, where $(\rho^{\tau_B})_{ij,kl}=(\rho)_{il,kj}$. Any state that violates the PPT criterion is an entangled one. Previous works have shown that entanglement and coherence, as two types of quantum correlations, have many similarities. Thus we consider whether partial transposition can be applied to coherence detection.

Now, consider the bipartite scenario. With respect to the partially transposed matrix $\rho^{\tau_B}$ of $\rho$, the PT moments are defined as
\begin{equation}\label{e2}
	T_k=\text{Tr}[(\rho^{\tau_B})^k],~~~k=1,2,...,MN,
\end{equation}
then one have the following results.
\begin{theorem}
	PT moments of an incoherent state $\rho$ is invariant under partial transposition:
	\begin{equation}\label{e3}
		\text{Tr}(\rho^k)=\text{Tr}[(\rho^{\tau_B})^k],~~~k=1,2,...,MN.
	\end{equation}
\end{theorem}
\begin{proof}
	The characteristic polynomial of $\rho^{\tau_B}$ is considered by us
	\begin{equation}\label{e4}
		\xi_0\lambda^{MN}-\xi_1\lambda^{MN-1}+...+(-1)^{MN}\xi_{MN},
	\end{equation}
	where $\xi_0=1$, $\xi_k=\sum_{\{s_k\}\in S}\prod_{j\in s_k}\lambda^{'}_j$, $k=1,2,...,MN,$ ${s_k}$ denotes a subset of $S=\{1,2,...,MN\}$ with $k$ elements and $\lambda^{'}_j\,(j=1,...,MN)$ represent the eigenvalues of $\rho^{\tau_B}$. Note that the characteristic polynomial coefficients and
	the PT moments have the following relations \cite{jfjq},
	\begin{equation}\label{e5}
		\xi_{k+1}=\frac1{k+1}\sum_{l=0}^{k}(-1)^l\xi_{k-l}T_{l+1}
	\end{equation}
	for $k=0,1,...,MN-1$. 
	
According to (\ref{e4}) and (\ref{e5}), it is proved that the PT moments of $\rho$ are invariant, which is equivalent to proving that the eigenvalues of $\rho$ are invariant after partial transposition. Therefore, one first prove that the basis-independence feature of the spectrum for partially transposed state. In fact, the different choices of basis to perform partial transposition can be captured with a local unitary operation $U_B$ acting on Bob’s side,
\begin{equation*}
	\begin{split}
		\lambda[(U_B\rho\,U_B^{\dagger})^{\tau_B}]&=\lambda[(U_B\rho\,U_B^{\dagger})^{\tau_A}]\\
		&=\lambda(U_B\rho^{\tau_A}U_B^{\dagger})\\
		&=\lambda(\rho^{\tau_A})\\
		&=\lambda(\rho^{\tau_B}),
	\end{split}
\end{equation*}
where $\rho^{\tau_A}$ obtained from partial transpose with respect to subsystem $A$. The first equation is due to the fact that the spectrum remains unchanged under usual transposition and $\lambda(\rho^{\tau_B})=\lambda[(\rho^{\tau_A})^\tau]$, where $\tau$ specifies the usual transposition acting on the joint system. An incoherent state has the form of (\ref{e1}) under the computational basis $\mathcal{B}$, so its spectrum remains unchanged under partial transposition. We complete the proof of Theorem 1.
\end{proof}
According to Theorem 1, if any PT moment of a quantum state changes under partial transposition, the state is coherent, and full quantum tomography is not required to detect coherence.  In addition, it is worth noting that PT moments have the advantage that they can be estimated using shadow tomography \cite{ekhb} in a more efficient way than if one had to reconstruct $\rho$ via full quantum state tomography. In other words, our criterion can also be experimentally implemented.

\mathbi{Example 1}.

The family of quantum states with two parameters is considered by us,
$$
\sigma=\frac{1}{4}\left(\begin{array}{ccccccccc}
	1 & 0 & 0 & 0 & 0 & 0 & 0 & 0 & a \\
	0 & b & \frac{b}{2} & 0 & 0 & 0 & 0 & 0 & 0 \\
	0 & \frac{b}{2} & b & 0 & 0 & 0 & 0 & 0 & 0\\
	0 & 0 & 0 & 0 & 0 & 0 & 0 & 0 & 0 \\
	0 & 0 & 0 & 0 & 1-b & 0 & b & 0 & 0 \\
	0 & 0 & 0 & 0 & 0 & 0 & 0 & 0 & 0 \\
	0 & 0 & 0 & 0 & b & 0 & 1-b & 0 & 0 \\
	0 & 0 & 0 & 0 & 0 & 0 & 0 & 0 & 0 \\
	a & 0 & 0 & 0 & 0 & 0 & 0 & 0 & 1
\end{array}\right),
$$
with $0\leq a\leq1$ and $0\leq b\leq\frac{1}{2}$.

In fact, our justification begins with the third moment, which is due to the fact that partial transposition is trace-preserving and $\text{Tr}(\rho^2)=T_2$ is trivially satisfied. By using Theorem 1, we can obtain that our criterion can detect all coherent states in this state family, and note that $\text{Tr}(\sigma^3)-T_3=0$ if and only if $a=b=0$, which means that one do not even need to compare all the moments, sometimes the first few moments may be enough to complete the task. Please refer to Figure 1 for details.
\begin{figure}[htbp]
	\centering
	\includegraphics[width=0.5\textwidth]{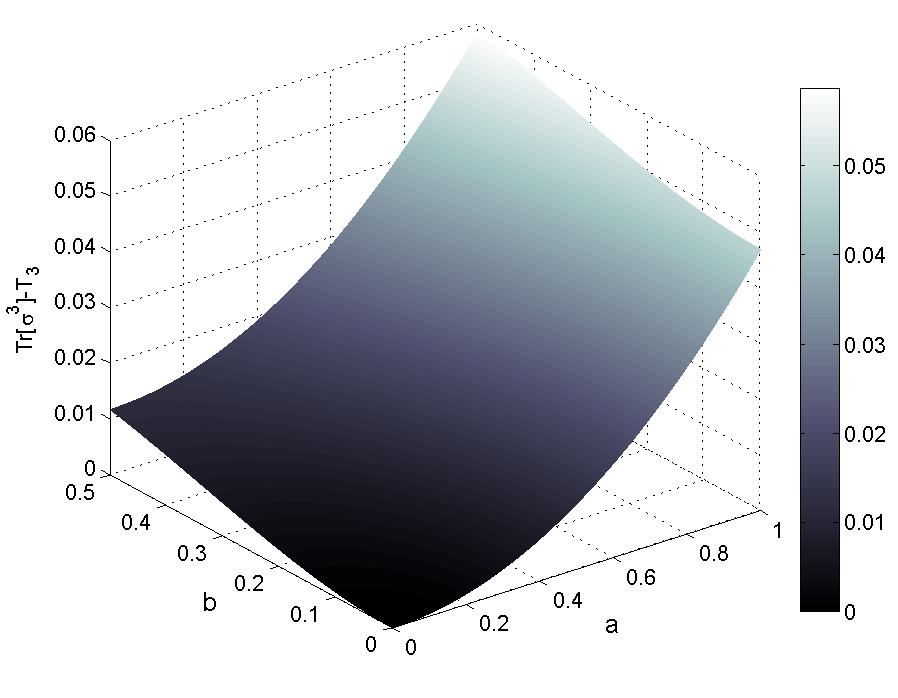}
	\vspace{-2em} \caption{The relationship between $\text{Tr}(\sigma^3)-T_3$ and two parameters for $\sigma$ in example 1 is shown in the Figure.} \label{Fig.1}
\end{figure}

\section{III. Nonlinear detection of coherence from witnesses}
A coherence witness is a Hermitian operator $W$ such that $\text{Tr}(W\rho_{co})<0$ for a coherent state $\rho_{co}$ and $\text{Tr}(W\rho_{IC})\geq0$ for all incoherent states $\rho_{IC}$. In addition, there is another type of coherence witness, which works in the following way: $\text{Tr}(W\rho_{co})\neq0$ for a coherent state $\rho_{co}$ and $\text{Tr}(W\rho_{IC})=0$ for all incoherent states $\rho_{IC}$.

Similar to the linear detection of entanglement witness, the above method can also be called the linear detection of coherence witness. The authors introduced a method called nonlinear entanglement detection in \cite{rh}, where the linear detection failed, but the nonlinear measurement considered by the author may succeed. So a natural idea is to apply this nonlinear detection to coherence witness. We present a nonlinear detection method for coherence and its detection strategy is similar to the linear counterpart: for a nonlinear coherence witness $\mathcal{W}$, we also have $\text{Tr}(\mathcal{W}\rho_{IC}^{\otimes k})=0$ for all incoherent states, while $\text{Tr}(\mathcal{W}\rho_{co}^{\otimes k})\neq 0$ for a coherent state $\rho_{co}$.

Like entangled counterpart \cite{rh}, we also require nonlinear witness $\mathcal{W}$ to take the form of the tensor product of witness. Its working way can be clearly explained by the following simple example. For a bipartite state $\rho_{AB}$, one take the tensor product of two bipartite witness,
$$
\mathcal{W}=W_{A_1A_2}\otimes V_{B_1B_2}.
$$
Here, $A_1$ and $B_1$ denote the first and second qubits of the quantum state, respectively, where $A_2$ and $B_2$ correspond to the first and second qubits of the copy state. Then the expectation value
$$
\langle\mathcal{W}\rangle_{\rho^{\otimes2}}=\text{Tr}[(W_{A_1A_2}\otimes V_{B_1B_2})\rho_{A_1B_1}\otimes\rho_{A_2B_2}]
$$
is $0$ if $\rho$ is incoherent. See Figure 2. The observation in the following shows that our strategy can indeed work.
\begin{figure}[htbp]
	\centering
	\includegraphics[width=0.35\textwidth]{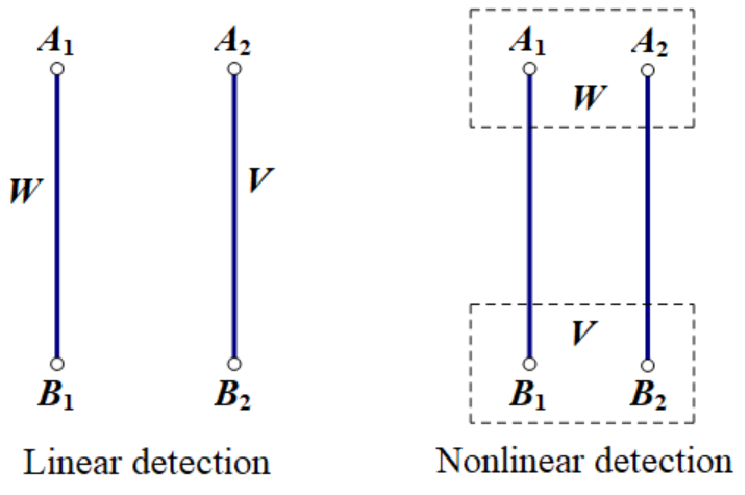}
	\vspace{-0.5em} \caption{The ideas of linear detection and nonlinear detection are shown in the Figure.} \label{Fig.2}
\end{figure}

\begin{observation}
	There exist a coherent state $\rho$ and witness $W$ and $V$ such that $\text{Tr}(W\rho)=0$ and $\text{Tr}(V\rho)=0$, but $\text{Tr}[(W_{A_1A_2}\otimes V_{B_1B_2})\rho^{\otimes2}]\neq0$.
	
	In simple terms, linear detection fails, but nonlinear coherence detection with $W\otimes V$ succeeds.
\end{observation}

In the previous section, our coherence criterion cannot detect the $X$ state with the same anti-diagonal terms. Now using our nonlinear detection to consider a more general class of $X$ states.

\mathbi{Example 2}.

$$
\rho_{x}=\frac{1}{4}\left(\begin{array}{cccc}
	1 & 0 & 0 & \alpha \\
	0 & 1 & \beta & 0 \\
	0 & \beta & 1 & 0 \\
	\alpha & 0 & 0 & 1 \\
\end{array}\right),
$$
where $0\leq\alpha,\beta\leq1$. The witnesses we adopt are
\begin{equation*}
	\begin{split}
		W=&|00\rangle\langle10|+|10\rangle\langle00|,\\
		V=&|01\rangle\langle11|+|11\rangle\langle01|.
	\end{split}
\end{equation*}
Obviously, the coherence of $\rho_x$ cannot be detected by $W$ and $V$, i.e. $\text{Tr}(W\rho_x)=\text{Tr}(V\rho_x)=0$. However, our strategy shows that $\text{Tr}[(W_{A_1A_2}\otimes V_{B_1B_2})\rho_x^{\otimes2}]=\frac{\alpha+\beta}{8}$, see Appendix A. When we apply the coherence witnesses $W$ and $V$, even if they are not the optimal coherent witness for this $X$ state, we can still detect their coherence based on our nonlinear detection.

Furthermore, although the expectation values need to be calculated in a way that matches the dimensions of $\mathcal{W}$ and $\rho^{\otimes k}$, the dimensions of witness $W$ and $\rho$ can be different.

\mathbi{Example 3}.

Take the 2-qubit witness
$$
W=\frac{1}{2}(X\otimes X+Y\otimes Y)
$$
and the GHZ state affected by white noise $\rho_g=\frac{1-g}{8}\mathbbm{1}+g|GHZ\rangle\langle GHZ|$, where $|GHZ\rangle=\frac{1}{\sqrt{2}}(|000\rangle+|111\rangle)$.

By calculation, it is not difficult to obtain that $\text{Tr}(W^{\otimes3}\rho_g^{\otimes2})=0$. However, by utilizing our nonlinear method, one have
$$
\text{Tr}[(W_{A_1A_2}\otimes W_{B_1B_2}\otimes W_{C_1C_2})\rho_g^{\otimes2}]=\frac{1}{2}\neq0.
$$
Thus $\rho_g$ becomes detectable using our strategy with the coherence witness $W$. The calculation process is detailed in Appendix B.

In addition to the non-linear detection method mentioned above, we can also consider another non-linear measurement method, as shown in Figure 3.
\begin{figure}[htbp]
	\centering
	\includegraphics[width=0.35\textwidth]{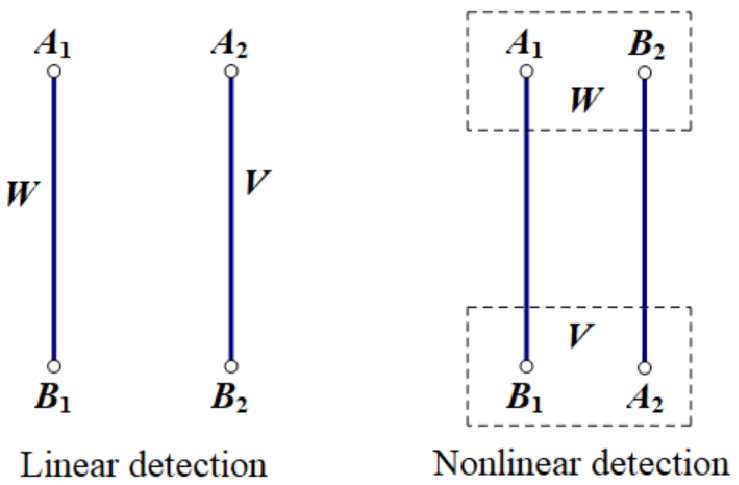}
	\vspace{-0.5em} \caption{The another nonlinear detection method is shown in the Figure.} \label{Fig.3}
\end{figure}

It is not difficult to verify that for example 2, we have $\text{Tr}[(W_{A_1B_2}\otimes V_{B_1A_2})\rho_x^{\otimes2}]=\frac{\alpha+\beta}{8}$, and for example 3, we have 
$$
\text{Tr}[(W_{A_1B_2}\otimes W_{B_1C_2}\otimes W_{C_1A_2})\rho_g^{\otimes2}]=\frac{1}{2}\neq0.
$$
In other words, this non-linear detection method achieves the same effect as the former. However, the following observation and example may demonstrate an advantage of the latter.

\begin{observation}
	There exists a coherent state $\rho_{co}$ and coherence witnesses $W_i\,(i=1,2,3)$, where both linear detection and nonlinear detection of two copies fail for $\rho_{co}$, i.e. 
	$$
	\text{Tr}(W_i\rho_{co})=\text{Tr}[(W_{i}\otimes W_{j})\rho_{co}^{\otimes2}]=0,\,i,j=1,2,3.
	$$
	However we have
	$$
	\text{Tr}[(W_{i}\otimes W_{j}\otimes W_k)\rho_{co}^{\otimes3}]\neq0
	$$
	for some $i,j,k\in\{1,2,3\}$.
	
	In simple terms, the nonlinear detection of two copies fails, but three copies succeeds.
\end{observation}

\mathbi{Example 4}.

Let us consider the tripartite quantum state 
$$
\rho_c=\frac{1}{2}\left(\begin{array}{cccccccc}
	c & 0 & 0 & 0 & 0 & 0 & 0 & c \\
	0 & 1-c & 0 & 0 & 0 & 0 & 0 & 0 \\
	0 & 0 & 0 & 0 & 0 & 0 & 0 & 0 \\
	0 & 0 & 0 & 0 & 0 & 0 & 0 & 0 \\
	0 & 0 & 0 & 0 & 0 & 0 & 0 & 0 \\
	0 & 0 & 0 & 0 & 0 & 0 & 0 & 0 \\
	0 & 0 & 0 & 0 & 0 & 0 & 1-c & 0 \\
	c & 0 & 0 & 0 & 0 & 0 & 0 & c
\end{array}\right),
$$
where $0<c\leq1$, and the coherence witness is taken as $W=|001\rangle\langle111|+|111\rangle\langle001|$.

After calculation, it is not difficult to verify that 
$$
\text{Tr}(W\rho_c)=\text{Tr}[(W_{A_1B_2C_2}\otimes W_{B_1C_1A_2})\rho_c^{\otimes2}]=0.
$$
But when we consider the nonlinear detection of three copies, one have
$$
\text{Tr}[(W_{A_1B_2C_3}\otimes W_{B_1C_2A_3}\otimes W_{C_1A_2B_3})\rho_c^{\otimes3}]\neq0.
$$
Details can be found in Appendix C.

What does this example tell us? If linear detection fails, we can consider nonlinear detection of two copies; If the nonlinear detection of two copies fails, we can consider three copies and so on. This provides us with a new perspective and a novel path for our nonlinear idea. This approach can be seen in Figure 4.
\begin{figure}[htbp]
	\centering
	\includegraphics[width=0.45\textwidth]{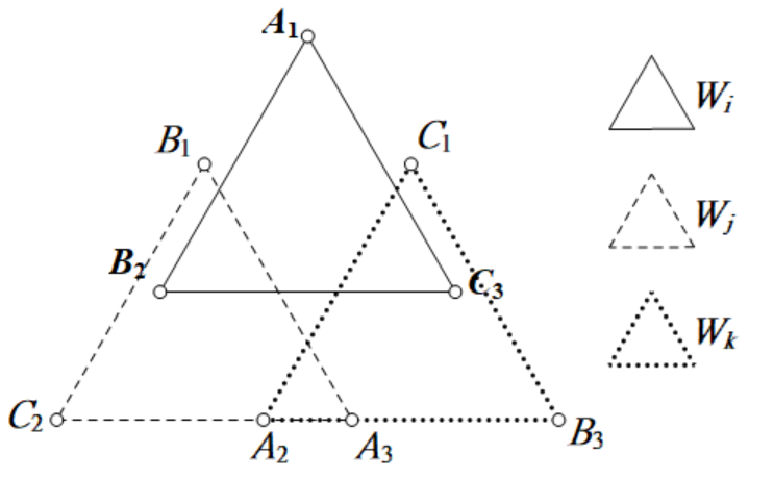}
	\vspace{-1em} \caption{The nonlinear detection of three copies is shown in the Figure.} \label{Fig.4}
\end{figure}
\section{IV. Lower bound of coherence measure from witness}

In addition to coherence detection, the task of quantifying coherence has also emerged as one of the prominent themes of quantum theory. Analogous to the setting of entanglement \cite{vt}, the coherence of a quantum state can be quantified by its robustness to the noise \cite{ctmm},
$$
\mathcal{R}(\rho)=\min\limits_{\sigma}R(\rho\|\sigma).
$$
We refer to $R(\rho\|\sigma)$ as the robustness of $\rho$ relative to $\sigma$, which is the minimal $s$ such that $\rho(s)=\frac{1}{1+s}(\rho+s\sigma)$ is incoherent. Note that $\sigma$ here is an arbitrary state. This is slightly different from the entanglement situation, for which the original robustness was defined in terms of pseudomixtures with separable states.

Using the distance between a given state and the incoherent state set to measure coherence is another direct and concise method. Various coherence measures can be generated by selecting various distances. A possible measure to compare two quantum states is the trace distance \cite{mc}
$$
\mathcal{D}(\rho,\sigma)=\frac{1}{2}\text{Tr}|\rho-\sigma|,
$$
where $\text{Tr}|\varrho|=\text{Tr}(\varrho\varrho^{\dagger})^{\frac{1}{2}}$ is the trace norm or 1-norm of $\varrho$. Corresponding to the trace distance, there is a coherence measure defined as the trace distance between a given state $\rho$ and the incoherent state set IC:
$$
\mathcal{E}(\rho)=\min\limits_{\sigma\in IC}\mathcal{D}(\rho,\sigma).
$$
Now, our idea is to seek the connection between the coherence measures $\mathcal{R}(\rho)$ and $\mathcal{E}(\rho)$, and investigate whether these possible connections imply a lower bound of $\mathcal{R}(\rho)$.
\begin{theorem}
	Given a quantum coherent state $\rho$, we have
	\begin{equation}\label{e6}
		\mathcal{R}(\rho)\geq\frac{\mathcal{E}(\rho)}{1-\mathcal{E}(\rho)}.
	\end{equation}
\end{theorem}
\begin{proof}
	Firstly, we set the optimal state achieving the minimal $s$ as $\sigma_{opt}$, that is to say, $\rho(s)=\frac{1}{1+s}(\rho+s\sigma_{opt})\in$ IC. Then one have
	\begin{equation*}\nonumber
		\begin{split}
			\frac{s}{1+s}&=\frac{s}{1+s}. 1\\
			&\geq \frac{s}{1+s}. \frac{\text{Tr}|\rho-\sigma_{opt}|}{2}\\
			&=\frac{s}{1+s}.\frac{1+s}{s}.\frac{\text{Tr}|\rho-\rho(s)|}{2}\\
			&=\mathcal{D}(\rho,\rho(s))\\
			&\geq\mathcal{E}(\rho).
		\end{split}
	\end{equation*}
	The first inequality is due to the fact that for any two states $\rho_1$ and $\rho_2$, we have $\text{Tr}|\rho_1-\rho_2|\leq2$, and the last inequality can be immediately obtained through the definition of $\mathcal{E}$. The inequality above is equivalent to
	$$
	1+\frac{1}{s}\leq \frac{1}{\mathcal{E}(\rho)}.
	$$
	Thus we have 
	$$
	\mathcal{R}(\rho)=s\geq\frac{\mathcal{E}(\rho)}{1-\mathcal{E}(\rho)}.
	$$
\end{proof}
Based on the fact that $y=\frac{x}{1-x}$ is an increasing function, the lower bound of $\mathcal{R}(\rho)$ can be obtained by considering the lower bound of $\mathcal{E}(\rho)$. Note that there is an inequality relationship between the trace norm of matrix $A$ and $|\text{Tr}(A)|$: $\text{Tr}|A|\geq|\text{Tr}(A)|$. Thus, we have the following results.
\begin{lemma}
	For a witness operator $\mathcal{W}$ with $-\mathbbm{1}\leq\mathcal{W}\leq\mathbbm{1}$, we have
	\begin{equation}\label{e7}
		\mathcal{E}(\rho)\geq\frac{1}{2}|\text{Tr}((\rho-\sigma_{opt})\mathcal{W})|,
	\end{equation}
	where $\sigma_{opt}$ represents the incoherent state closest to $\rho$ in the sense of trace distance.
\end{lemma}
\begin{proof}
	We use $(\alpha_1,\alpha_2,...,\alpha_n)$ and $(\beta_1,\beta_2,...,\beta_n)$ to represent the singular values of $\rho-\sigma_{opt}$ and $\mathcal{W}$ arranged in descending order, respectively, and $n$ is the order of both of them. Then one have
	\begin{equation}
		\begin{split}
			\mathcal{E}(\rho)&=\min\limits_{\sigma\in IC}\frac{1}{2}\text{Tr}|\rho-\sigma|\nonumber\\
			&=\frac{1}{2}\text{Tr}|\rho-\sigma_{opt}|\\
			&=\frac{1}{2}\sum_{i=1}^{n}\alpha_i\\
			&\geq\frac{1}{2}\sum_{i=1}^{n}\alpha_i\beta_i\\
			&\geq\frac{1}{2}|\text{Tr}((\rho-\sigma_{opt})\mathcal{W})|.
		\end{split}
	\end{equation}
	The first inequality is due to $0\leq\beta_i\leq1$ $(i=1,...,n)$, which is because $\mathcal{W}$ is a Hermitian matrix and the condition $-\mathbbm{1}\leq\mathcal{W}\leq\mathbbm{1}$. The second inequality is obtained by the majorization relationship and the relationship between the trace norm of matrix $A$ and $|\text{Tr}(A)|$: $\text{Tr}|A|\geq|\text{Tr}(A)|$.
\end{proof}

It is clear that if we want to make $\mathcal{W}$ satisfy the condition $-\mathbbm{1}\leq\mathcal{W}\leq\mathbbm{1}$, then $\mathcal{W}$ needs to be normalized. And based on the (\ref{e7}), different normalization methods will lead to the different lower bounds of $\mathcal{E}(\rho)$.
A simple and natural normalization method is to divide $\mathcal{W}$ by its maximum singular value. However, the lower bound given by this method is not very good here. The authors in \cite{szt} provide an inspiring idea
\begin{equation}\label{e8}
	\mathcal{W}_N=\frac{2\mathcal{W}-(\lambda_++\lambda_-)\mathbbm{1}}{\lambda_+-\lambda_-}.
\end{equation}
Where $\lambda_+$ and $\lambda_-$ represent the maximum and minimum eigenvalues of $\mathcal{W}$, respectively. Now, we show that (\ref{e8}) is a method of normalizing a given witness operator $\mathcal{W}$. 
\begin{lemma}
	The witness operator $\mathcal{W}_N$ defined by (\ref{e8}) satisfies 
	$$
	-\mathbbm{1}\leq\mathcal{W}_N\leq\mathbbm{1}.
	$$
\end{lemma}
\begin{proof}
	Using $\lambda$ to represent an eigenvalue of $\mathcal{W}$, based on (\ref{e8}), the corresponding eigenvalue of $\mathcal{W}_N$ is
	\begin{equation}\label{e9}
		\lambda_N=\frac{2\lambda-(\lambda_++\lambda_-)}{\lambda_+-\lambda_-}.
	\end{equation}
	To prove that $-\mathbbm{1}\leq\mathcal{W}_N\leq\mathbbm{1}$, we only need to prove that any eigenvalue of $\mathcal{W}_N$ falls between -1 and 1. According to the linear relationship in Eq. (\ref{e9}), we only need to consider the case when $\lambda=\lambda_+$ and the case when $\lambda=\lambda_-$.
	\begin{equation}\nonumber
		\begin{split}
			\lambda_{N,+}&=\frac{2\lambda_+-(\lambda_++\lambda_-)}{\lambda_+-\lambda_-}\\
			&=\frac{2\lambda_+-\lambda_+-\lambda_-}{\lambda_+-\lambda_-}\\
			&\leq 1.
		\end{split}
	\end{equation}
	Meanwhile, $\lambda_{N,-}\geq-1$ can be similarly proven. Therefore, it can be concluded that
	$$
	-\mathbbm{1}\leq\mathcal{W}_N\leq\mathbbm{1}.
	$$
\end{proof}
We introduce the lower bound of $\mathcal{E}(\rho)$ in the following.
\begin{theorem}
	If a coherent state $\rho$ can be detected by the witness $\mathcal{W}$, then we have
	\begin{equation}\label{e10}
		\mathcal{E}(\rho)\geq L_{\mathcal{W}_N}\equiv\frac{-\text{Tr}(\rho\mathcal{W})}{\lambda_+-\lambda_-},
	\end{equation}
	where $\lambda_+$ and $\lambda_-$ are the maximum eigenvalue and minimum eigenvalue of $\mathcal{W}$, respectively.
\end{theorem}
\begin{proof}
	Firstly, normalize $\mathcal{W}$ to obtain $\mathcal{W}_N$. Based on the Lemma 1, one have
	\begin{small}
		\begin{eqnarray*}
			\mathcal{E}(\rho)&\geq&\frac{1}{2}|\text{Tr}((\rho-\sigma_{opt})\mathcal{W}_N)|\\
			&=&\frac{1}{2}\frac{|2.\text{Tr}((\rho-\sigma_{opt})\mathcal{W})|}{\lambda_+-\lambda_-}\\
			&=&\frac{1}{2}\frac{2.|\text{Tr}(\rho\mathcal{W})-\text{Tr}(\sigma_{opt}\mathcal{W})|}{\lambda_+-\lambda_-}\\
			&\geq&\frac{-\text{Tr}(\rho\mathcal{W})}{\lambda_+-\lambda_-}\\
			&=&L_{\mathcal{W}_N}.
		\end{eqnarray*}
	\end{small}
	The first equation is derived from a simple application of the fact that the density operator trace is 1, and the second inequality is natural because a non negative term is discarded by us.
\end{proof}
Then we can immediately obtain the following results.
\begin{corollary}
		For a coherent quantum state $\rho$, we have
		\begin{equation}\label{e11}
			\mathcal{R}(\rho)\geq\frac{L_{\mathcal{W}_N}}{1-L_{\mathcal{W}_N}},
		\end{equation}
		where the expression for $L_{\mathcal{W}_N}$ is given by Eq.\,(10).
\end{corollary}

We will test the effectiveness of our lower bound through the following examples.

\mathbi{Example 5}.

Consider the $d$-dimensional maximum entangled state $|\psi^+\rangle=\frac{1}{\sqrt{d}}\sum_{i}|ii\rangle$ and the coherence witness is given by \cite{zzy},
$$
W=\bigtriangleup(|\psi^+\rangle\langle\psi^+|)-|\psi^+\rangle\langle\psi^+|,
$$
where $\bigtriangleup$ represents dephasing operation. The lower bound of coherence robustness given by \cite{zzy} is 
$$
\mathcal{R}(|\psi^+\rangle)\geq \min\limits_{\sigma\in IC}\||\psi^+\rangle\langle\psi^+|-\sigma\|_{2}^2=1-\frac{1}{d},
$$
where $\|M\|_2=\sqrt{\text{Tr}(M^{\dagger}M)}$ is called the 2-norm or Frobenius norm.

For the witness operator $W$, we have $\lambda_+=\frac{1}{d}$ and $\lambda_-=\frac{1}{d}-1$. By substituting these quantities into our lower bound (\ref{e10}), we obtain that $L_{\mathcal{W}_N}=1-\frac{1}{d}$. According to the Theorem 4, one have $\mathcal{R}(|\psi^+\rangle)\geq d-1$, which is a tight lower bound. 

\mathbi{Example 6}.

Let us consider the isotropic state,
$$
\rho_{iso}=\frac{1-v}{8}(\mathbbm{1}-|\Phi^+\rangle\langle\Phi^+|)+v|\Phi^+\rangle\langle\Phi^+|,
$$
where $0\leq v\leq1$ and $|\Phi^+\rangle=\frac{1}{\sqrt{3}}\sum_{i=0}^{2}|ii\rangle$. The coherence witness operator we take as $W_v=\bigtriangleup(\rho_{iso})-\rho_{iso}$.

Let us first compute our lower bound. The calculation yields $L_{\mathcal{W}_N}=\frac{(9v-1)^2}{|108v-12|}$. Hence, by Corollary 1, the lower bound is obtained as $L_R$:
\begin{equation*}
L_R=\frac{\frac{(9v-1)^2}{|108v-12|}}{1-\frac{(9v-1)^2}{|108v-12|}}.
\end{equation*}
To verify the validity of our lower bound, let us compare it against existing robustness lower bounds.
In \cite{ctmm}, the authors proposed a method for estimating $\mathcal{R}$,
\begin{equation}\label{e12}
	\begin{split}
		\mathcal{R}(\rho)&\geq L_1=-\text{Tr}(\rho\mathcal{W}),
	\end{split}
\end{equation}
where $-\mathbbm{1}\leq\mathcal{W}\leq\mathbbm{1}$ and $\mathcal{W}$ is a coherence witness. Thus on have
\begin{equation*}
L_1=-\text{Tr}(\rho_{iso}\mathcal{W}_N)=\frac{|1-3v|}{2}.
\end{equation*}
Here, $\mathcal{W}_N$ denotes the coherence witness operator obtained by substituting $W_v$ into Eq.\,(8).
The another lower bound of coherence robustness provided by \cite{zzy} is 
\begin{equation}\label{e13}
	\mathcal{R}(\rho_{iso})\geq L_2=\min\limits_{\sigma\in IC}\|\rho_{iso}-\sigma\|_2^2=\frac{(9v-1)^2}{96}.
\end{equation}
 The relationship between these lower bounds can be found in Figure 5.
\begin{figure}[htbp]
	\centering
	\includegraphics[width=0.5\textwidth,height=0.4\textwidth]{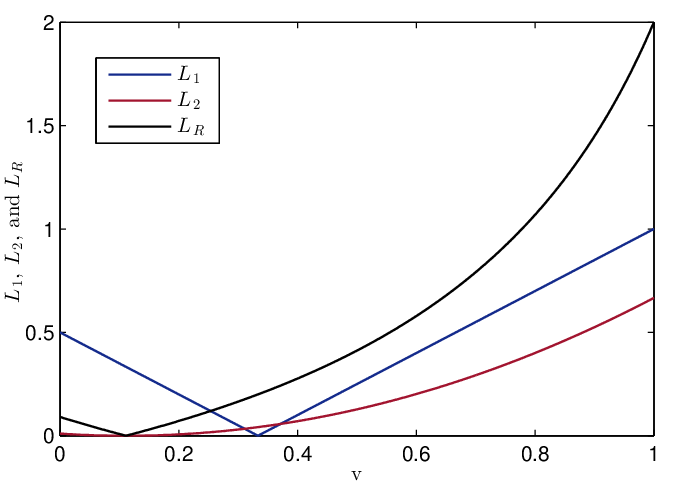}
	\vspace{-1em} \caption{The relationship between $L_1$, $L_2$ and our lower bound $L_R$ for $\rho_{iso}$ in example 6 is shown in the Figure.} \label{Fig.5}
\end{figure}
We find that for $v<0.254$, $L_1\geq L_R\geq L_2$, while for $v>0.254$, $L_R>\max\{L_1, L_2\}$. This indicates that our estimate of $\mathcal{R}(\rho_{iso})$ for $v>0.254$ is tighter than both $L_1$ and $L_2$. Furthermore, it is noteworthy that our lower bound $L_R$ satisfies $L_R\geq L_2$ for all $v\in[0,1]$. However, as observed in Figure 5, $L_2>L_1$ within a subinterval of $[0.3,0.4]$. While our bound demonstrates overall superior performance, integrating multiple bounds in practical applications may provide a tighter estimate of the coherence robustness $\mathcal{R}$.

\section{V. Conclusions and discussions}

In this paper, we investigated the detection and quantization for quantum coherence. We firstly provided a coherence criterion without full quantum state tomography. By detailed example we illustrated the effectiveness of these criterion in detecting coherence. Secondly, we proposed a nonlinear method for detecting coherence, where we found that for given coherent states and witnesses, normal witness detection fails but nonlinear detection succeeds. In particular, the dimensions of witness and state can be different. Moreover, we illustrated that when the nonlinear detection of the two copies of state fails, the three copies may be successful. This provides a novel perspective for our coherence detection. Robustness of coherence is a useful measure for coherence, but for general coherent states, coherent robustness is difficult to calculate. To address this, we provided a lower bound for coherence robustness from the witnesses. We found that our method provided tight lower bounds for some quantum states. We compared our lower bounds with two existing bounds, $L_1$ \cite{ctmm} and $L_2$ \cite{zzy}, in specific examples. Our bound was lower than $L_1$ in certain parameter ranges but higher than both bounds elsewhere, while $L_1$ itself was surpassed by $L_2$ in some intervals. Overall, our lower bound demonstrated superior performance. However, we concluded that combining these bounds in practical applications would be necessary to achieve a more robust estimation of coherence robustness.

As an important quantum correlation, coherence plays a pivotal role in various quantum technologies. Our findings can deepen our understanding of the methodologies for coherence detection and quantization. 

\bigskip Acknowledgments: This work is supported by the National Natural Science Foundation of China (NSFC) under Grant No.12204137; the Hainan Provincial Natural Science Foundation of China under Grant No.125RC744, the specific research fund of the Innovation Platform for Academicians of Hainan Province.

\begin{widetext}
	\appendix
	\section{Appendix A}
In example 2, we consider a more general class of $X$ states
$$
\rho_{x}=\frac{1}{4}\left(\begin{array}{cccc}
	1 & 0 & 0 & \alpha \\
	0 & 1 & \beta & 0 \\
	0 & \beta & 1 & 0 \\
	\alpha & 0 & 0 & 1 \\
\end{array}\right)
$$
and the witnesses we set are $W=|00\rangle\langle10|+|10\rangle\langle00|$ and $V=|01\rangle\langle11|+|11\rangle\langle01|$.

Now utilizing our nonlinear detection. For simplicity, we will write down each nonzero term of $\rho_x^{\otimes2}$ after the action of observable $W_{A_1B_2}$ and $V_{B_1A_2}$. Without causing confusion, we will abbreviate $|0\rangle_{A_1}|0\rangle_{B_1}|0\rangle_{A_2}|0\rangle_{B_2}$ as $|0000\rangle$, and so on for other terms.
\begin{equation*}
\begin{split}
&|0001\rangle\langle0001|\xrightarrow[V_{B_1B_2}]{W_{A_1A_2}}|1101\rangle\langle0001|,\hspace{2.5cm}|0001\rangle\langle0010|\xrightarrow[V_{B_1B_2}]{W_{A_1A_2}}|1101\rangle\langle0010|,\\
\bigstar&|0001\rangle\langle1101|\xrightarrow[V_{B_1B_2}]{W_{A_1A_2}}|1101\rangle\langle1101|,\hspace{2.5cm}|0001\rangle\langle1110|\xrightarrow[V_{B_1B_2}]{W_{A_1A_2}}|1101\rangle\langle1110|,\\
&|0101\rangle\langle0101|\xrightarrow[V_{B_1B_2}]{W_{A_1A_2}}|1001\rangle\langle0101|,\hspace{2.5cm}|0101\rangle\langle0110|\xrightarrow[V_{B_1B_2}]{W_{A_1A_2}}|1001\rangle\langle0110|,\\
\bigstar&|0101\rangle\langle1001|\xrightarrow[V_{B_1B_2}]{W_{A_1A_2}}|1001\rangle\langle1001|,\hspace{2.5cm}|0101\rangle\langle1010|\xrightarrow[V_{B_1B_2}]{W_{A_1A_2}}|1001\rangle\langle1010|,\\
\bigstar&|1001\rangle\langle0101|\xrightarrow[V_{B_1B_2}]{W_{A_1A_2}}|0101\rangle\langle0101|,\hspace{2.5cm}|1001\rangle\langle0110|\xrightarrow[V_{B_1B_2}]{W_{A_1A_2}}|0101\rangle\langle0110|,\\
&|1001\rangle\langle1001|\xrightarrow[V_{B_1B_2}]{W_{A_1A_2}}|0101\rangle\langle1001|,\hspace{2.5cm}|1001\rangle\langle1010|\xrightarrow[V_{B_1B_2}]{W_{A_1A_2}}|0101\rangle\langle1010|,\\
\end{split}
\end{equation*}
\begin{equation*}
	\begin{split}
		\bigstar&|1101\rangle\langle0001|\xrightarrow[V_{B_1B_2}]{W_{A_1A_2}}|0001\rangle\langle0001|,\hspace{2.5cm}|1101\rangle\langle0010|\xrightarrow[V_{B_1B_2}]{W_{A_1A_2}}|0001\rangle\langle0010|,\\
		&|1101\rangle\langle1101|\xrightarrow[V_{B_1B_2}]{W_{A_1A_2}}|0001\rangle\langle1101|,\hspace{2.5cm}|1101\rangle\langle1110|\xrightarrow[V_{B_1B_2}]{W_{A_1A_2}}|0001\rangle\langle1110|.
	\end{split}
\end{equation*}
Thus we have 
$$
\text{Tr}[(W_{A_1B_2}\otimes V_{B_1A_2})\rho_x^{\otimes2}]=\frac{\alpha}{16}+\frac{\beta}{16}+\frac{\beta}{16}+\frac{\alpha}{16}=\frac{\alpha+\beta}{8}.
$$
\section{Appendix B}
Now, we use our strategy to detect the GHZ state affected by white noise $\rho_g=\frac{1-g}{8}\mathbbm{1}+g|GHZ\rangle\langle GHZ|$ given in example 3. For simplicity, we also only list nonzero terms in the following.
\begin{equation}\nonumber
	\begin{split}
		&|000111\rangle\langle000000|\xrightarrow[W_{B_1B_2},W_{C_1C_2}]{W_{A_1A_2}}|111000\rangle\langle000000|,\hspace{1.5cm}|000111\rangle\langle000111|\xrightarrow[W_{B_1B_2},W_{C_1C_2}]{W_{A_1A_2}}|111000\rangle\langle000111|,\\
		\bigstar&|000111\rangle\langle111000|\xrightarrow[W_{B_1B_2},W_{C_1C_2}]{W_{A_1A_2}}|111000\rangle\langle111000|,\hspace{1.5cm}|000111\rangle\langle111111|\xrightarrow[W_{B_1B_2},W_{C_1C_2}]{W_{A_1A_2}}|111000\rangle\langle111111|,\\
		\bigstar&|111000\rangle\langle000111|\xrightarrow[W_{B_1B_2},W_{C_1C_2}]{W_{A_1A_2}}|000111\rangle\langle000111|,\hspace{1.5cm}|111000\rangle\langle000000|\xrightarrow[W_{B_1B_2},W_{C_1C_2}]{W_{A_1A_2}}|000111\rangle\langle000000|,\\
		&|111000\rangle\langle111000|\xrightarrow[W_{B_1B_2},W_{C_1C_2}]{W_{A_1A_2}}|000111\rangle\langle111000|,\hspace{1.5cm}|111000\rangle\langle111111|\xrightarrow[W_{B_1B_2},W_{C_1C_2}]{W_{A_1A_2}}|000111\rangle\langle111111|.
	\end{split}
\end{equation}
Thus one have
$$
\text{Tr}[(W_{A_1A_2}\otimes W_{B_1B_2}\otimes W_{C_1C_2})\rho_g^{\otimes2}]=\frac{1}{2}\neq0.
$$

\section{Appendix C}

In this section, we consider the nonlinear detection of three copies of $\rho_c$ in example 4, that is to say, one caculate the values of $\text{Tr}[(W_{A_1B_2C_3}\otimes W_{B_1C_2A_3}\otimes W_{C_1A_2B_3})\rho_c^{\otimes3}]
$.
\begin{equation*}
	\begin{split}
		&|000000111\rangle\langle000000000|\xrightarrow[W_{B_1C_2A_3},W_{C_1A_2B_3}]{W_{A_1B_2C_3}}|111111111\rangle\langle000000000|,\\
		&|000000111\rangle\langle000000111|\xrightarrow[W_{B_1C_2A_3},W_{C_1A_2B_3}]{W_{A_1B_2C_3}}|111111111\rangle\langle000000111|,\\
		&|000000111\rangle\langle000111000|\xrightarrow[W_{B_1C_2A_3},W_{C_1A_2B_3}]{W_{A_1B_2C_3}}|111111111\rangle\langle000111000|,\\
		&|000000111\rangle\langle000111111|\xrightarrow[W_{B_1C_2A_3},W_{C_1A_2B_3}]{W_{A_1B_2C_3}}|111111111\rangle\langle000111111|,\\
		&|000000111\rangle\langle111000000|\xrightarrow[W_{B_1C_2A_3},W_{C_1A_2B_3}]{W_{A_1B_2C_3}}|111111111\rangle\langle111000000|,\\
		&|000000111\rangle\langle111000111|\xrightarrow[W_{B_1C_2A_3},W_{C_1A_2B_3}]{W_{A_1B_2C_3}}|111111111\rangle\langle111000111|,\\
		&|000000111\rangle\langle111111000|\xrightarrow[W_{B_1C_2A_3},W_{C_1A_2B_3}]{W_{A_1B_2C_3}}|111111111\rangle\langle111111000|,\\
		\bigstar&|000000111\rangle\langle111111111|\xrightarrow[W_{B_1C_2A_3},W_{C_1A_2B_3}]{W_{A_1B_2C_3}}|111111111\rangle\langle111111111|,\\
		&|111111111\rangle\langle000000000|\xrightarrow[W_{B_1C_2A_3},W_{C_1A_2B_3}]{W_{A_1B_2C_3}}|000000111\rangle\langle000000000|,\\
	\end{split}
\end{equation*}
\begin{equation*}
\begin{split}
	\bigstar&|111111111\rangle\langle000000111|\xrightarrow[W_{B_1C_2A_3},W_{C_1A_2B_3}]{W_{A_1B_2C_3}}|000000111\rangle\langle000000111|,\\
	&|111111111\rangle\langle000111000|\xrightarrow[W_{B_1C_2A_3},W_{C_1A_2B_3}]{W_{A_1B_2C_3}}|000000111\rangle\langle000111000|,\\
	&|111111111\rangle\langle000111111|\xrightarrow[W_{B_1C_2A_3},W_{C_1A_2B_3}]{W_{A_1B_2C_3}}|000000111\rangle\langle000111111|,\\
	&|111111111\rangle\langle111000000|\xrightarrow[W_{B_1C_2A_3},W_{C_1A_2B_3}]{W_{A_1B_2C_3}}|000000111\rangle\langle111000000|,\\
	&|111111111\rangle\langle111000111|\xrightarrow[W_{B_1C_2A_3},W_{C_1A_2B_3}]{W_{A_1B_2C_3}}|000000111\rangle\langle111000111|,\\
	&|111111111\rangle\langle111111000|\xrightarrow[W_{B_1C_2A_3},W_{C_1A_2B_3}]{W_{A_1B_2C_3}}|000000111\rangle\langle111111000|,\\
	&|111111111\rangle\langle111111111|\xrightarrow[W_{B_1C_2A_3},W_{C_1A_2B_3}]{W_{A_1B_2C_3}}|000000111\rangle\langle111111111|.
\end{split}
\end{equation*}
Thus when $c\neq0$, we have 
$$
\text{Tr}[(W_{A_1B_2C_3}\otimes W_{B_1C_2A_3}\otimes W_{C_1A_2B_3})\rho_c^{\otimes3}]=\frac{c^3}{2}+\frac{c^3}{2}=c^3\neq0.
$$
\end{widetext}
\end{document}